\documentclass[letterpaper, 10pt, conference]{ieeeconf}      

\IEEEoverridecommandlockouts                              
\usepackage{amsmath,graphicx,amsfonts,amssymb,epsfig,subfigure,mathrsfs}
\usepackage{color}
\usepackage{multirow}
\usepackage{rotating}
\usepackage{graphicx}%
\usepackage{algorithm}
\usepackage{algpseudocode}
\usepackage{algorithmicx}
\usepackage{cite}
\usepackage{framed}

\newtheorem{theorem}{Theorem}

\newtheorem{definition}{Definition}

\newtheorem{remark}{Remark}

\newtheorem{example}{Example}

\title{\LARGE \bf
Geodesic Density Tracking with Applications to Data Driven Modeling
}

\author{Abhishek Halder, Raktim Bhattacharya
\thanks{Abhishek Halder and Raktim Bhattacharya are with the Department of Aerospace Engineering, Texas A\&M University,
        College Station, TX 77843, USA,
        {\tt\small \{ahalder,raktim\}@tamu.edu}}%
}

\begin{document}

\maketitle
\thispagestyle{empty}
\pagestyle{empty}

\begin{abstract}
Many problems in dynamic data driven modeling deals with distributed rather than lumped observations. In this paper, we show that the Monge-Kantorovich optimal transport theory provides a unifying framework to tackle such problems in the systems-control parlance. Specifically, given distributional measurements at arbitrary instances of measurement availability, we show how to derive dynamical systems that interpolate the observed distributions along the geodesics. We demonstrate the framework in the context of three specific problems: (i) \emph{finding a feedback control} to track observed ensembles over finite-horizon, (ii) \emph{finding a model} whose prediction matches the observed distributional data, and (iii) \emph{refining a baseline model} that results a distribution-level prediction-observation mismatch. We emphasize how the three problems can be posed as variants of the optimal transport problem, but lead to different types of numerical methods depending on the problem context. Several examples are given to elucidate the ideas.
\end{abstract}


\section{Introduction}
In traditional systems theory, modeling and control synthesis assumes the availability of measurements in the form of vector \emph{signals} or \emph{trajectories} observed over time. However, in many applications, observations are not lumped variables, rather they are distributed over spatial dimensions. For example, in Nuclear Magnetic Resonance (NMR) spectroscopy and Imaging (MRI) applications, the measurement variable is magnetization distribution \cite{LiKhaneja2009}, since sensing individual magnetization states of the order of Avogadro number $6 \times 10^{23}$, remains a technological limitation. On the other hand, in process industry applications like paper-making \cite{Wang1999,WangBakiKabore2001}, measurement and control of distributions are motivated by the design choice of tracking desired fibre length and filler size distributions. Similar examples can be found in biological systems \cite{BrownMoehlisHolmes2004}.

Another motivation to consider modeling, identification and control problems in the distributional setting, comes from the recent proliferation of cyberphysical systems. The tight integration of control, communication and computation has resulted information deluge, popularly termed as ``big data" \cite{BigDataRef}. With the abundance of data, it becomes imperative to seek constructive algorithms that can lead to better phenomenological models or better controllers, specially in the presence of parametric uncertainties and lack of understanding of the first-principle physics. The objective of this paper is to introduce a framework, grounded on the theory of optimal transport \cite{Villani2003}, to generate models that track distributional observations over finite horizon.

\subsection{Nomenclature and the Interpretation of Probability}
We use the terms ``distribution" and ``ensemble" interchangeably with the common meaning of the availability of a multitude of real measurements at a fixed time. Also, we assume that the underlying true dynamics generating the measurements is smooth enough \cite{DiPernaLions1989} so that the realizations are absolutely continuous, and hence we can talk about ``densities" in lieu of ``distributions". The optimal transport framework described in this paper will work even if this assumption is violated \cite{Villani2003}, namely if the distributions exist but the densities don't.

Notice that the observations are in distributional or density level, does not \emph{necessarily} imply that the underlying state dynamics is governed by a partial differential equation (PDE) or stochastic differential equation (SDE). Indeed, the density may arise from the parametric and initial condition uncertainties, although the underlying state dynamics may be deterministic, governed by an (unknown) ordinary differential equation (ODE). In this sense, the term ``density" refers to the concentration of trajectories, and since the state space mass is preserved under the action of the flow, one can interpret \cite{LasotaMackey} the trajectory concentration as \emph{probability density} in the \emph{sense of propensity} \cite{JCWillems2010}. On the other hand, it may indeed be the case that the underlying true dynamics generating the observations, is governed by a SDE, naturally giving rise to probability densities, even if the initial conditions and parameters are fixed. Also, in cases where the observable is naturally distributed over spatial dimensions (e.g. a color image), then one can suitably normalize the data to enable our density-based optimal transport framework. Hence, without loss of generality, we term the observations as \emph{probability density functions} (PDFs).

\subsection{Related Work}
The idea of PDF control is not new in the control literature. Previous works like \cite{ForbesACC2003, ZhuPEM2012, GuoWangAutomatica2005} dealt with \emph{asymptotic} density shaping, meaning the feedback controls were found so that the desired PDF coincides with the stationary PDF of the closed-loop dynamics. Similar ideas predate in terms of covariance control \cite{SkeltonCDC1985}. The framework presented in this paper differs from the existing literature in that we track the distributions observed over \emph{finite} horizons, and the time instances of distributional measurement availability \emph{need not be equi-spaced}. The question then becomes: \emph{``over any given horizon, what needs to be done at the realization level trajectory dynamics, such that we track the observations at the ensemble level in some optimal sense?"} The optimal transport theory allows us to achieve finite-time distributional tracking while guaranteeing that \emph{minimum amount of work is done over each horizon}.

\subsection{Notations}
Most notations are standard. The symbol $\sharp$ denotes push-forward of a probability measure. $\ker\left(.\right)$ and $\text{Im}\left(.\right)$ refer to the kernel and image of a linear operator, respectively. The symbol ``$\text{Id}$" denotes
identity vector map of appropriate dimension. The notation $x \sim \rho$ means that the random vector $x$ has the joint PDF $\rho\left(x\right)$. Furthermore, $\nabla$ denotes the gradient operator with respect to spatial variables, and $\text{det}\left(\cdot\right)$ refers to determinant of a matrix. The symbol $\text{Hess}\left(\cdot\right)$ stands for the Hessian. Unless otherwise specified, the superscript $^{\star}$ refers to optimality, while the superscript $^{\dagger}$ refers to the Moore-Penrose pseudo-inverse of a matrix. The notation $I$ denotes the identity matrix, $\mathcal{N}\left(\mu,\Sigma\right)$ denotes Gaussian PDF with mean $\mu$ and covariance $\Sigma$, $\mathcal{U}\left(\cdot\right)$ denotes uniform PDF.

\begin{figure}[tb]
\begin{center}
 \includegraphics[scale=0.77]{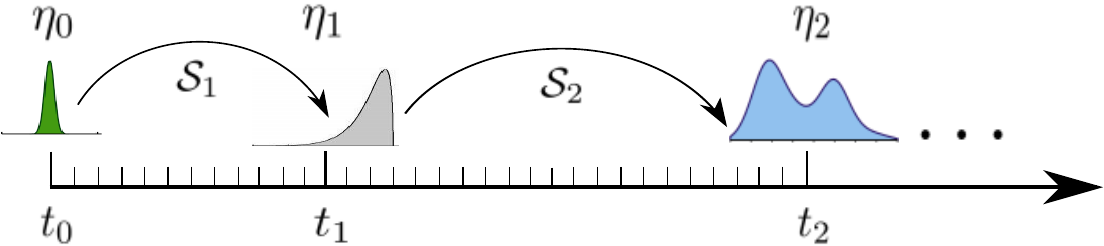}
 \end{center}
 \vspace*{-0.1in}
\caption{The schematic of Problem 1, where a sequence of joint PDFs $\eta_{j}$ of the output vector $y_{j}$ are given at times $t_{j}$, $j=0,1,\hdots,M$. The objective is to find dynamical systems $S_{j+1} : y_{j} \mapsto y_{j+1}$, over each horizon $[t_{j}, t_{j+1})$.}
\label{OverallSchematic}
\vspace*{-0.1in}
\end{figure}
\vspace*{-0.05in}

\subsection{Problem Formulation}
We consider a sequence of time instances $\{t_{j}\}_{j=0}^{M}$, when the measurement vector $y\left(t\right) \in \mathbb{R}^{d}$ is recorded as a sequence of PDFs $\{\eta_{j} \triangleq \eta\left(y\left(t_{j}\right), t_{j}\right)\}_{j=0}^{M}$. If we introduce $y_{j} \triangleq y\left(t_{j}\right)$, then we have $y_{j} \sim \eta_{j}$. The general problem statement can now be stated as follows (see Fig. \ref{OverallSchematic}).
\begin{framed}
\noindent%
\textbf{Problem 1}\\
\noindent%
For each interval $[t_{j}, t_{j+1})$, find a dynamical system $\mathcal{S}_{j+1} : y_{j} \mapsto y_{j+1}$ that minimizes the total transportation cost
\begin{eqnarray}
\displaystyle\int_{\mathbb{R}^{2d}} \parallel y_{j+1} - y_{j} \parallel_{\ell_{2}\left(\mathbb{R}^{d}\right)}^{2} \: \rho\left(y_{j}, y_{j+1}\right) \: dy_{j} dy_{j+1},
\label{TransportCost}
\end{eqnarray}
over all transportation policy $\rho\left(y_{j},y_{j+1}\right)$, such that $y_{j} \sim \eta_{j}$, $y_{j+1} \sim \eta_{j+1}$, $j = 0, 1, \hdots, M$.
\end{framed}
To understand the problem, let $dm_{j \rightarrow j+1} \triangleq \rho\left(y_{j}, y_{j+1}\right) \: dy_{j} dy_{j+1}$, which is nothing but the differential mass over the product space $\mathbb{R}^{2d}$. If the cost per unit mass equals squared Euclidean distance $\parallel y_{j+1} - y_{j} \parallel_{\ell_{2}\left(\mathbb{R}^{d}\right)}^{2}$, then (\ref{TransportCost}) denotes the total cost to transport the density $\eta_{j}$ to $\eta_{j+1}$, while preserving mass (since both $\eta_{j}$ and $\eta_{j+1}$ are PDFs). Notice that the total cost depends on the choice of the PDF $\rho\left(y_{j},y_{j+1}\right)$ supported over $\mathbb{R}^{2d}$, that dictates the transportation policy, and hence $\mathcal{S}_{j+1}$. However, finding the joint PDF $\rho\left(y_{j},y_{j+1}\right)$ with given marginals $\eta_{j}$ and $\eta_{j+1}$, is not unique. Thus, we seek to find that transportation policy or joint PDF $\rho\left(y_{j},y_{j+1}\right)$, which is the minimizer of the total transportation cost (\ref{TransportCost}).

Notice that \textbf{Problem 1} is rather generic in the sense, it does not impose any structural constraint on the dynamical system $\mathcal{S}_{j+1}$ to be determined. In the context of dynamic data driven modeling, we are interested in three variants of \textbf{Problem 1}, stated next.
\begin{framed}
\noindent%
\textbf{Problem 1.1 (Finite Horizon Feedback Control of Output PDFs)}\\
\noindent%
Solve Problem 1 with \emph{pre-specified control structure} on $\mathcal{S}_{j+1}$ (e.g. affine or non-affine, linear or nonlinear) by finding the state or output feedback control $u\left(\cdot\right)$.
\end{framed}
\begin{remark}
Notice that \textbf{Problem 1.1} is more specific than \textbf{Problem 1} in the sense that although $\mathcal{S}_{j+1}$ could be found by solving Problem 1, additional conditions may be necessary for feedback control $u\left(\cdot\right)$ to exist that satisfies the desired control structure.
\end{remark}
\begin{framed}
\noindent%
\textbf{Problem 1.2 (Data Driven $d$\textsuperscript{th} Order Modeling)}\\
\noindent%
Solve Problem 1 with no \emph{a priori} knowledge about $\mathcal{S}_{j+1}$ other than \emph{the type of temporal dependence} in $[t_{j},t_{j+1})$ (e.g. continuous-time flow or discrete-time map).
\end{framed}
\begin{remark}
The order of modeling/identification is same as the dimension ($d$) of the output vector. This is helpful in practice since $d$ is usually much less than the dimension of the \emph{true} state space. In other words, the model (ODE or map) will be over the outputs, thus naturally resulting reduced order models.
\end{remark}
\begin{framed}
\noindent%
\textbf{Problem 1.3 (Model Refinement)}\\
\noindent%
Solve Problem 1 with $\mathcal{S}_{j+1}$ as an instantaneous map, the source PDF $\eta_{j}$ as the nominal model prediction, and the target PDF $\eta_{j+1}$ as the true measurement. Here $\eta_{j}$ and $\eta_{j+1}$ are given at the same physical time. The refined model is a composition of $\mathcal{S}_{j+1}$ with the output map of the model.
\end{framed}
\begin{remark}
In \textbf{Problem 1.3}, $t \in [t_{j},t_{j+1})$ refers to a synthetic notion of time. In the context of refinement problem, the physical time stays zero order hold for each time-horizon in Fig. \ref{OverallSchematic}.
\end{remark}

\subsection{Organization of the Paper}
This paper is structured as follows. In Section II, the necessary background on optimal transport theory is given. In particular, we connect different formulations of optimal transport with the different dynamic data driven modeling problems described in the previous subsection, and show how they lead to different types of numerical solutions. Next, using the ideas from Section II, we discuss \textbf{Problem} \textbf{1.1}, \textbf{1.2} and \textbf{1.3} in Sections III, IV and V, respectively. To illustrate the solution methodology, analytical and numerical results are provided for example problems. Section VI concludes the paper.


\section{Background on Optimal Transport}

\subsection{Primal Formulation}
The optimal transport theory originated in 1781 when Gaspard Monge considered \cite{Monge1781} the problem of moving a pile of soil from an excavation to another site that entails minimum work. This idea went mostly unnoticed for 160 years until Leonid Kantorovich provided a modern treatment \cite{Kantorovich1942} of this subject in 1942 (the English translation \cite{Kantorovich1958} appeared in 1958), which eventually led to the Nobel prize in economics in 1975. In the theory of Monge-Kantorovich optimal transport, one defines a distance, called \emph{Wasserstein distance}, between two \emph{given} PDFs $\rho$ and $\widehat{\rho}$, that measures the \emph{shape difference} between them.
\begin{definition} (\textbf{Wasserstein distance})
The $L_{2}$ Wasserstein distance of order 2 (henceforth referred simply as \emph{Wasserstein distance} $W$), between two $d$-dimensional random vectors $y \sim \rho$, and $\widehat{y} \sim \widehat{\rho}$, is defined as
\begin{eqnarray}
W\left(\rho,\widehat{\rho}\right) \triangleq \left(\underset{\varrho \in \mathcal{P}_{2}\left(\rho,\widehat{\rho}\right)}{\text{inf}} \mathbb{E}\left[ \parallel y - \widehat{y} \parallel_{\ell_{2}\left(\mathbb{R}^{d}\right)}^{2} \right]\right)^{\frac{1}{2}},
\label{WassDefn}
\end{eqnarray}
where the $\mathbb{E}\left[\cdot\right]$ is taken with respect to the joint PDF $\varrho\left(y,\widehat{y}\right)$ that makes the cost function achieve the infimum. The symbol $\mathcal{P}_{2}\left(\rho,\widehat{\rho}\right)$ denotes the set of all joint PDFs supported over $\mathbb{R}^{2d}$, having finite second moments, whose first marginal is $\rho$, and second marginal is $\widehat{\rho}$.
\end{definition}
\begin{remark}
It can be shown \cite{Rachev1991} that $W$ defines a metric on the manifold of PDFs, and remains well defined between the \emph{distributions} even though the random vectors $y$ and $\widehat{y}$ are not absolutely continuous (i.e. $\rho$ and $\widehat{\rho}$ don't exist).
\end{remark}
\begin{remark}
From (\ref{WassDefn}), notice that $W^{2}$ is nothing but the \emph{minimum mean square error (MMSE)} between two random vectors, and is equal to the \emph{least transportation cost} (\ref{TransportCost}).
\end{remark}
\begin{remark}
In (\ref{WassDefn}), the cost function
\begin{eqnarray}
J_{1}\left(\varrho\right) \triangleq \int_{\mathbb{R}^{2d}} \parallel y - \widehat{y}\parallel_{2}^{2} \varrho\left(y,\widehat{y}\right) dyd\widehat{y}
\label{WassCost}
\end{eqnarray}
and the constraints: $\int_{\mathbb{R}^{d}} \varrho\left(y,\widehat{y}\right) \: d\widehat{y} = \rho\left(y\right)$, $\int_{\mathbb{R}^{d}} \varrho\left(y,\widehat{y}\right) \: dy = \widehat{\rho}\left(\widehat{y}\right)$, $\varrho\left(y,\widehat{y}\right) \geq 0$, are linear in the function $\varrho$. Hence, computing $W$ from (\ref{WassDefn}) requires solving an \emph{infinite-dimensional linear program}. As stated in Section I.D, the infimizer $\varrho^{\star}$ results the optimal transportation plan.
\end{remark}
\begin{remark}
The infinite dimensional LP (\ref{WassDefn}) can be solved by directly discretizing the problem in terms of the samples of the constituent PDFs (see \cite{HalderBhattacharya2011, HalderBhattacharya2012}). As shown in the first row of Table \ref{TheBigTable}, this results a large scale finite dimensional LP, whose solution provides a consistent approximation \cite{Villani2003} of the true solution (of the infinite dimensional LP).
\end{remark}

\subsection{Variational Formulation for Optimal Transport Map}
Instead of solving (\ref{WassDefn}), one could directly solve for the the optimal transport map $\beta : \mathbb{R}^{d} \mapsto \mathbb{R}^{d}$, that satisfies $y = \beta\left(\widehat{y}\right)$, by solving
{\small{\begin{eqnarray}
\underset{\beta\left(\cdot\right)}{\text{inf}} \underbrace{\displaystyle\int_{\mathbb{R}^{d}} \parallel \beta\left(\widehat{y}\right) - \widehat{y} \parallel_{\ell_{2}\left(\mathbb{R}^{d}\right)}^{2} \: \widehat{\rho}\left(\widehat{y}\right) \: d\widehat{y}}_{J_{2}\left(\beta\right)}, \; \text{subject to} \; \rho = \beta \:\sharp\: \widehat{\rho}.
\label{BetaVariational}
\end{eqnarray}}}
\begin{remark}
Since there are infinite ways to morph $\widehat{\rho}$ to $\rho$, (\ref{BetaVariational}) looks for an \emph{optimal push-forward map} $\beta^{\star}\left(\cdot\right)$ that would require minimum amount of transport effort among all possible push-forward maps $\beta\left(\cdot\right)$. Then the map $\beta^{\star}\left(\cdot\right)$ characterizes the optimal transport of \textbf{Problem 1}.
\end{remark}
\begin{remark}
In a seminal paper \cite{brenier1991polar}, Brenier proved the existence and uniqueness of $\beta^{\star}\left(\cdot\right)$. Further, his \emph{polar factorization theorem} \cite{brenier1991polar} proved that the unique vector function $\beta^{\star}\left(\cdot\right)$ can be written as a gradient of a scalar function, i.e. $\beta^{\star} = \nabla \psi$. Furthermore, the scalar function $\psi$ is convex. The optimal transport map $\beta^{\star}$ is also known as the \emph{Brenier map}.
\label{BrenierPolarFactorizationRemark}
\end{remark}
\begin{remark}
Although the cost function in (\ref{BetaVariational}) is quadratic in $\beta\left(\cdot\right)$, the push-forward constraint is nonlinear and non-convex in $\beta\left(\cdot\right)$. Thus, a direct numerical optimization is not straight-forward. As shown in the second row of Table \ref{TheBigTable}, \cite{TannenbaumSISC2010} used the fact that $\beta^{\star}\left(\cdot\right)$ is curl-free, to formulate a \emph{regularized} sequential quadratic program (SQP) to solve (\ref{BetaVariational}) as
\begin{eqnarray}
\underset{\beta\left(\cdot\right)}{\inf} \: \widetilde{J}_{2}\left(\beta\right), \quad \text{subject to} \quad c\left(\beta\right) = 0,
\end{eqnarray}
where $\widetilde{J}_{2}\left(\beta\right) \triangleq J_{2}\left(\beta\right) + \alpha \displaystyle\int_{\mathbb{R}^{d}} \parallel \nabla \times \beta \parallel_{\ell_{2}\left(\mathbb{R}^{d}\right)}^{2} \: d\widehat{y}$, and $\alpha > 0$ is a regularization parameter. The constraint $c\left(\beta\right) \triangleq \lvert \text{det}\left(\nabla \beta\right) \rvert \, \rho \circ \beta\left(\widehat{y}\right) - \widehat{\rho}\left(\widehat{y}\right)$.
\end{remark}

\subsection{PDE Formulation for Optimal Transport Map}
From Remark \ref{BrenierPolarFactorizationRemark}, we can substitute $\beta = \nabla \psi$ in the push-forward constraint $c\left(\beta\right) = 0$. Then it follows that $\psi$ must solve
\begin{eqnarray}
\lvert \text{det}\left(\text{Hess}\left(\psi\left(\widehat{y}\right)\right)\right) \rvert \: \rho\left(\nabla \psi\left(\widehat{y}\right)\right) = \widehat{\rho}\left(\widehat{y}\right).
\label{MongeAmperePDE}
\end{eqnarray}
This is a second order, nonlinear, stationary, elliptic PDE, known as the \emph{Monge-Amp\`{e}re equation}, to be solved for $\psi$ as a function of $\widehat{y}$. In principle, if we can solve (\ref{MongeAmperePDE}), then $\nabla \psi$ would solve \textbf{Problem 1}. However, as mentioned in the third row of Table \ref{TheBigTable}, numerically solving the PDE (\ref{MongeAmperePDE}) remains a research challenge.

\subsection{Benamou-Brenier Space-time Variational Formulation}
Benamou and Brenier proposed \cite{benamou2000computational} a \emph{dynamic} reformulation of the \emph{static} optimization problem (\ref{BetaVariational}) by introducing a synthetic notion of time, which we denote as $s \in [0, \tau]$. Their main result is that the spatial optimization problem (\ref{BetaVariational}), is equivalent to solving the following space-time optimization problem:
\begin{eqnarray}
W^{2} = \tau \underset{\left(\varphi, v\right)}{\text{inf}} \: \displaystyle\int_{\mathbb{R}^{d}} \int_{0}^{\tau} \: \varphi\left(\widehat{y}, s\right) \parallel v\left(\widehat{y}, s\right) \parallel_{\ell_{2}\left(\mathbb{R}^{d}\right)}^{2} \: d\widehat{y}\:ds, \label{BBcost}\\
{\small{\text{subject to}}} \: \displaystyle\frac{\partial \varphi}{\partial s} + \nabla \cdot \left(\varphi v\right) = 0, \: \varphi\left(\cdot, 0\right) = \widehat{\rho}, \: \varphi\left(\cdot, \tau\right) = \rho.
\label{BenamouBrenierEulerian}
\end{eqnarray}

\begin{remark}
It is important to understand the meaning of solving the optimization problem (\ref{BBcost})-(\ref{BenamouBrenierEulerian}). Notice that the spatial and temporal integrals in the cost function can be interchanged. Thus, if we fix $s$, then the cost is the instantaneous kinetic energy of the ensemble during transport, where each sample moves according to the deterministic ODE $\displaystyle\frac{d}{ds}\widehat{y} = v\left(\widehat{y}(s), s\right)$, corresponding to the Liouville PDE \cite{HalderBhattacharyaJGCDLiouville} $\displaystyle\frac{\partial \varphi}{\partial s} + \nabla \cdot \left(\varphi v\right) = 0$, appearing in the constraint. Hence, the cost function in (\ref{BBcost}) is equal to the total kinetic energy up to time $\tau$. Consequently, $W^{2}$ equals total work done during the transport process. The optimization is over a pair of vector field $v$ and joint PDF $\varphi$, and is convex in both.
\end{remark}

\begin{remark}
It can be shown \cite{Villani2003} that the minimizing vector field $v^{\star}\left(\widehat{y}, s\right)$ in the above optimization problem, is a \emph{pressureless potential flow}. In other words, $v^{\star}\left(\widehat{y}, s\right) = \nabla \phi\left(\widehat{y}, s\right)$, where the scalar function $\phi\left(\widehat{y}, s\right)$ solves the Hamilton-Jacobi equation
\begin{eqnarray}
\displaystyle\frac{\partial \phi}{\partial s} + \displaystyle\frac{1}{2} \parallel \nabla \phi \parallel_{\ell_{2}\left(\mathbb{R}^{d}\right)}^{2} \;=\; 0.
\end{eqnarray}
\label{GradientFlow}
\end{remark}

\begin{remark}
In p. 384 of \cite{benamou2000computational}, using Legendre transform, (\ref{BBcost})-(\ref{BenamouBrenierEulerian}) was further converted to a saddle point optimization problem, which was numerically solved using the augmented Lagrangian technique \cite{fortin1983augmented}. Recently, an improved numerical method to solve (\ref{BBcost})-(\ref{BenamouBrenierEulerian}) has been proposed \cite{GabrielProxOpSplitArXiv2013} via proximal operator splitting. We will use this technique in Section IV for numerical simulations.
\end{remark}


\begin{center}
\begin{table*}[t]
\centering
{\small{
\begin{tabular}{ | c | c | c | c | }
\hline \hline
 & & & \\
Mathematical formulation & Problem type & Numerical method & Solver\\
& & & \\ \hline \hline
& & & \\
Primal formulation in $\varrho$ & Infinite dimensional LP (\ref{WassDefn}) & Direct discretization & Large scale LP solver\\
& & & (e.g. MOSEK\textsuperscript{\textregistered} used in \cite{HalderBhattacharya2011, HalderBhattacharya2012}) \\
& & & \\ \hline
& & & \\
Variational Formulation for & Quadratic cost with nonlinear & Regularized SQP & ``Discretize-then-optimize" \\
optimal transport map $\beta\left(\cdot\right)$ & non-convex constraints (\ref{BetaVariational}) & (optimality not guaranteed) & solver in \cite{TannenbaumSISC2010}\\
& & & \\ \hline
& & & \\
PDE Formulation for & Monge-Amp\`{e}re PDE (\ref{MongeAmperePDE}) & Not well studied & See review article \cite{SIAMreview2013}\\
optimal transport map $\beta\left(\cdot\right)$ &  (Second order nonlinear elliptic PDE) & & for current research status\\
& & & \\ \hline
& & & \\
Brenier-Benamou space-time & Non-smooth convex & Proximal operator splitting & Staggered grid\\
variational formulation in $\left(v, \varphi\right)$ &  optimization problem (\ref{BBcost})-(\ref{BenamouBrenierEulerian}) & & discretization in \cite{GabrielProxOpSplitArXiv2013}\\
& & & \\ \hline
\hline
\end{tabular}
}}
\vspace*{0.1in}
\caption{Different formulations of \textbf{Problem 1} in Section I.D, as described in Section II.}
\label{TheBigTable}
\end{table*}
\end{center}


\subsection{Wasserstein Geodesics on the Manifold of PDFs}
One merit of the Benamou-Brenier approach described in Section II.D is that it constructs the transportation path, which is a geodesic connecting the source and target PDFs, and yields the intermediate PDFs satisfying McCann's displacement interpolation \cite{McCann1997}. In particular, the following results \cite{Villani2003} hold.

\begin{enumerate}
\item Without loss of generality, let the synthetic time $s \in \left[0,1\right]$, i.e. in the notation of Section II.D, set $\tau = 1$. Then the Benamou-Brenier vector field constructs the geodesic curve between $\left(\widehat{\rho}, \rho\right) : s\in\left[0,1\right] \mapsto \varphi\left(s\right)$. Recall that the PDF $\widehat{\rho}$ is the source PDF, and $\rho$ is the target PDF. In other words, $\varphi$ has the variational characterization
    \begin{eqnarray}
    \varphi\left(s\right) = \underset{\varphi}{\text{argmin}} \left(1-s\right) W^{2}\left(\widehat{\rho},\varphi\right) \:+\: s W^{2}\left(\rho,\varphi\right),
    \end{eqnarray}
    and it lies on the geodesic curve connecting $\widehat{\rho}$ and $\rho$. The Wasserstein distance $W\left(\widehat{\rho},\rho\right)$ is the length of this geodesic curve on the manifold of PDFs.\\

\item As a corollary of the above result, the intermediate optimal transport map $\beta_{s}$ that satisfies $\varphi\left(s\right) = \beta_{s} \:\sharp\: \widehat{\rho}$, is obtained via \emph{linear interpolation} between the identity map $\text{Id}$ and $\beta$, i.e.
    \begin{eqnarray}
    \beta_{s} = \left(1 - s\right) \text{Id} + s \beta.
    \label{MapIsLinearInterpolation}
    \end{eqnarray}
    Also, the intermediate Wasserstein distance is obtained via \emph{linear interpolation}:
    \begin{eqnarray}
    W\left(\widehat{\rho},\varphi\left(s\right)\right) &=& s \: W\left(\widehat{\rho},\rho\right), \label{WIsLinearInterpolation1}\\
    W\left(\rho,\varphi\left(s\right)\right) &=& \left(1 - s\right) \: W\left(\widehat{\rho},\rho\right).
    \label{WIsLinearInterpolation2}
    \end{eqnarray}
    However, the intermediate PDF is obtained via \emph{nonlinear (displacement) interpolation}:
    \begin{eqnarray}
    \varphi\left(s\right) &=& \beta_{s} \: \sharp \: \widehat{\rho} = \left[\left(1 - s\right) \text{Id} + s \beta\right] \: \sharp \: \widehat{\rho}, \label{PDFIsNonlinearInterpolation1}\\
    &=& \beta_{1-s} \: \sharp \: \rho = \left[s\:\text{Id} + (1 - s) \beta\right] \: \sharp \: \rho.
    \label{PDFIsNonlinearInterpolation2}
    \end{eqnarray}
\end{enumerate}

\section{Feedback Control for Finite-Horizon Density Tracking}

In this Section, we consider \textbf{Problem 1.1} under pre-specified control structures. For brevity, we only consider the case when $\mathcal{S}_{j+1}$ has discrete-time LTI state dynamics, and the state PDFs are prescribed Gaussians.

\subsection{Linear Gaussian PDF Control in Discrete Time}
Consider a linear system $x_{j+1} = A x_{j} + B u_{j}$, $x_{j} \in \mathbb{R}^{d}, u_{j} \in \mathbb{R}^{m}$, with a sequence of Gaussian PDFs $\eta_{j} = \mathcal{N}\left(\mu_{j}, \Sigma_{j}\right)$, $j=0,1,\hdots,M$. The objective is to find state feedback $u_{j}^{\star} \triangleq u^{\star}\left(x_{j}\right)$ over each time interval $\Delta t_{j} \triangleq [t_{j}, t_{j+1})$, such that $x_{j} \sim \eta_{j} = \mathcal{N}\left(\mu_{j}, \Sigma_{j}\right)$, while guaranteeing minimal transportation cost (\ref{TransportCost}). Using the idea of Section II.B, we transcribe the problem of finding optimal control $u_{j}^{\star}$ to that of finding the optimal transport map (a.k.a. \emph{Brenier map}) $\beta_{j}^{\star} : x_{j} \mapsto x_{j+1}$, where
\begin{eqnarray}
\beta_{j}^{\star} \triangleq \beta^{\star}\left(x_{j}\right) = \underset{\beta\left(.\right)}{\text{argmin}} \: \displaystyle\int_{\mathbb{R}^{d}}\parallel \beta\left(x_{j}\right) - x_{j} \parallel_{\ell_{2}\left(\mathbb{R}^{d}\right)}^{2} \eta_{j} \: dx_{j},
\label{DiscreteTimeSpatialOptimizationProblem}
\end{eqnarray}
subject to the constraints (C1) $x_{j} \sim \eta_{j}$, (C2) $\beta\left(x_{j}\right) \sim \eta_{j+1}$, and (C3) $\eta_{j+1} = \beta\:\sharp\:\eta_{j}$. Then we have the following result.

\begin{theorem}
Consider the discrete-time Gaussian PDF control problem under LTI structure, i.e. in Fig. \ref{OverallSchematic}, let $\eta_{j} = \mathcal{N}\left(\mu_{j},\Sigma_{j}\right)$, where $\ker\left(\Sigma_{j}\right)\:\cap\:\text{Im}\left(\Sigma_{j+1}\right) = \{0\}$. Further, let $\mathcal{S}_{j+1}$ be given by the discrete-time LTI structure: $x_{j+1} = A x_{j} + B u_{j}$, $\forall j=0,1,\hdots,M$. Then the state feedback $u_{j}^{\star} \triangleq u^{\star}\left(x_{j}\right)$ that minimizes the transportation cost (\ref{TransportCost}), has the following properties.
\begin{enumerate}
\item The optimal state feedback, if exists, must be affine.\\

\item Optimal state feedback $u_{j}^{\star}$ exists iff $\left(\Gamma_{j} - A\right), \gamma_{j} \in \text{ker}\left(I - BB^{\dagger}\right)$, where
    {\small{\begin{eqnarray}
    \Gamma_{j} &=& \sqrt{\Sigma_{j+1}} \left(\sqrt{\Sigma_{j+1}}\:\Sigma_{j}\:\sqrt{\Sigma_{j+1}}\right)^{-\frac{1}{2}} \sqrt{\Sigma_{j+1}}, \label{BrenierMapGauss2GaussMatrix}\\
    \gamma_{j} &=& \mu_{j+1} - \mu_{j}. \label{BrenierMapGauss2GaussVector}
    \end{eqnarray}}}

\item If exists, then the optimal state feedback is given by the pair $\left(K_{j},\kappa_{j}\right)$, i.e. $u_{j}^{\star} = K_{j} x_{j} + \kappa_{j}$, where $K_{j} = B^{\dagger}\left(\Gamma_{j} - A\right) - \left(I - BB^{\dagger}\right)R$, and $\kappa_{j} = B^{\dagger}\gamma_{j} - \left(I - BB^{\dagger}\right)r$, for \emph{arbitrary} real matrix-vector pair $\left(R,r\right)$ of appropriate dimensions.\\

\item If $B$ is full rank, then the optimal state feedback is unique, and is given by $K_{j} = B^{-1}\left(\Gamma_{j}-A\right)$, $\kappa_{j} = B^{-1}\gamma_{j}$.
\end{enumerate}
\end{theorem}
\vspace*{0.1in}
\begin{proof}
Given, $\ker\left(\Sigma_{j}\right)\:\cap\:\text{Im}\left(\Sigma_{j+1}\right) = \{0\}$, we know \cite{OlkinPukelsheim1982, KnottSmith1984} that $\beta_{j}^{\star}$ satisfying
\begin{eqnarray}
\mathcal{N}\left(\mu_{j+1}, \Sigma_{j+1}\right) = \beta_{j}^{\star} \:\sharp\: \mathcal{N}\left(\mu_{j}, \Sigma_{j}\right),
\end{eqnarray}
is a unique affine transformation $z \mapsto \Gamma_{j} z + \gamma_{j}$. Since the optimal transport is $x_{j+1} = \Gamma_{j} x_{j} + \gamma_{j}$, the \emph{optimal controller, if exists, must be affine}, i.e. of the form $u_{j}^{\star} = K_{j} x_{j} + \kappa_{j}$, where $K_{j}$ and $\kappa_{j}$ solve the \emph{linear matrix equations}
\begin{eqnarray}
A + B K_{j} = \Gamma_{j}, \qquad\quad B \kappa_{j} = \gamma_{j}.
\label{DiscreteTimeStateCovarianceController}
\end{eqnarray}
Now, from Lemma 2.4 in \cite{OharaKitamori1993}, there exists $K_{j}$ solving the equation $B K_{j} = \left(\Gamma_{j} - A\right)$ iff
\begin{eqnarray}
&B B^{\dagger} \left(\Gamma_{j} - A\right) = \left(\Gamma_{j} - A\right) \Leftrightarrow \left(I - BB^{\dagger}\right) \left(\Gamma_{j} - A\right) = 0 \nonumber\\
&\Leftrightarrow \left(\Gamma_{j} - A\right) \in \text{ker}\left(I - BB^{\dagger}\right).
\label{OharaKitamoriMatrixMatrixCondition}
\end{eqnarray}
On the other hand, the matrix-vector equation $B \kappa_{j} = \gamma_{j}$ admits solution iff
\begin{eqnarray}
BB^{\dagger}\gamma_{j} = \gamma_{j} \Leftrightarrow \gamma_{j} \in \text{ker}\left(I - BB^{\dagger}\right).
\label{MatrixVectorCondition}
\end{eqnarray}
When $\left(\Gamma_{j} - A\right), \gamma_{j} \in \text{ker}\left(I - BB^{\dagger}\right)$, then the (non-unique) solution is given by \cite{OharaKitamori1993}: $K_{j} = B^{\dagger}\left(\Gamma_{j} - A\right) - \left(I - BB^{\dagger}\right)R$, and $\kappa_{j} = B^{\dagger}\gamma_{j} - \left(I - BB^{\dagger}\right)r$, for \emph{arbitrary} real matrix-vector pair $\left(R,r\right)$ of appropriate dimensions. If $B$ is full rank, then $B^{-1}$ exists and $\left(I - BB^{\dagger}\right) = 0$, resulting the unique solution.
\end{proof}
\begin{remark}
Notice that till now, we assumed $\mathcal{S}_{j+1}$ is given by the same LTI pair $\left(A,B\right)$ $\forall j=0,1,\hdots,M$. It is easy to see that the above Theorem generalizes when the LTI pair $\left(A_{j},B_{j}\right)$ are different for different horizons.
\end{remark}

\subsection{How is this Different from Ensemble Control}
In recent years, a set of tools have been developed \cite{LiHarvardThesis2006, LiKhaneja2007, LiKhaneja2009} for finite-horizon distribution shaping using \emph{open-loop} control signal, under the constraint that the \emph{same open-loop excitation is applied to all realizations} of the system dynamics (characterized by parametric dispersions). Our approach of solving \textbf{Problem 1.1} is different from this paradigm, termed as ``ensemble control", in at least two ways. First, the ensemble control seeks an open-loop solution while ours is closed-loop solution. Second, the ensemble controllability needs to be established to guarantee the existence of such open loop excitation, however, if found, will allow sensor-less ensemble shaping. On contrary, our feedback implementation requires sensing at the ensemble level, but guarantees geodesic transport over each horizon.


\section{Data Driven $d$\textsuperscript{th} Order Modeling}
In this Section, we consider \textbf{Problem 1.2}, namely interpolating observed distributional data by identifying dynamical models over each finite horizon, in the absence of \emph{a priori} structural knowledge (unlike \textbf{Problem 1.1}) about the models. The only choice the modeler can make is to decide whether a discrete-time or continuous-time model is apt. Once this choice is made, a \emph{deterministic} trajectory-level model is desired that satisfies the two point boundary value problem in the output PDF level, at the beginning and end of the horizon length. Notice that we restrict ourselves to derive a deterministic flow or map, even though the observed PDFs may have been generated by a true but unknown state dynamics governed by PDE or SDE. In this sense, the modeling problem can be thought of as a sequence of finite-horizon distributional realization problems.

To illustrate how optimal transport ideas can befit here, we work out an example problem. Since the discrete-time Gaussian \emph{modeling} problem can be dealt similar to Section III.A, we choose a \emph{continuous-time non-Gaussian} scenario.

\begin{example}
Consider the case when the true dynamics is given by the Duffing oscillator
\begin{eqnarray}
\dot{x}_{1} = x_{2}, \quad \dot{x}_{2} = - \alpha x_{1}^{3} - \beta x_{1} - \delta x_{2}, \, y = \{x_{1}, \, x_{2}\}^{\top},
\label{DuffingDynamics}
\end{eqnarray}
where $\alpha = 1$, $\beta = -1$, $\delta = 0.5$. One can verify that for these values of the parameters $\alpha, \beta, \delta$, the dynamics (\ref{DuffingDynamics}) has three equilibria: $\left(0,0\right)$, $\left(\pm\sqrt{\displaystyle\frac{-\beta}{\alpha}}, 0\right)$. Linear stability analysis tells that the origin is a saddle node while the remaining two equilibria are stable foci. We use (\ref{DuffingDynamics}) only to generate synthetic data and assume that the knowledge of this true vector field is unavailable to the modeler.

To generate the true distributional PDFs, we assume that the initial joint state PDF $\xi_{0}\left(x_{0}\right) = \mathcal{U}\left(\left[-2,2\right]^{2}\right)$. We generate 500 samples from this uniform PDF, and evaluate them at $\xi_{0}$. Starting from these samples, we evolve the joint state PDF $\xi\left(x_{1}(t), x_{2}(t), t\right)$ subject to (\ref{DuffingDynamics}) by solving the Liouville PDE $\frac{\partial\xi}{\partial t} + \nabla\cdot\left(\xi f\right) = 0$, where $f\left(x_{1},x_{2}\right)$ is the Duffing vector field. We perform this uncertainty propagation by solving the method-of-characteristics ODE corresponding to the Liouville PDE. Details on this methodology can be found in Section II of \cite{HalderBhattacharyaJGCDLiouville}. This procedure results scattered colored data (Fig. \ref{DuffingScatteredData}) at every time $t_{j}$, $j = 1,2,\hdots,10$, where the location of the samples are determined from the dynamics while the color value at a sample location indicates the \emph{exact} (unlike Monte Carlo histograms) joint PDF value at that sample location, at that time. Since $y = \{x_{1}, \, x_{2}\}^{\top}$, hence Fig. \ref{DuffingScatteredData} depicts the sequence $\{t_{j}, \eta_{j}\}_{j=1}^{10}$, in the nomenclature of Section I.

\begin{figure}[tb]
\begin{center}
 \includegraphics[width=0.47\textwidth]{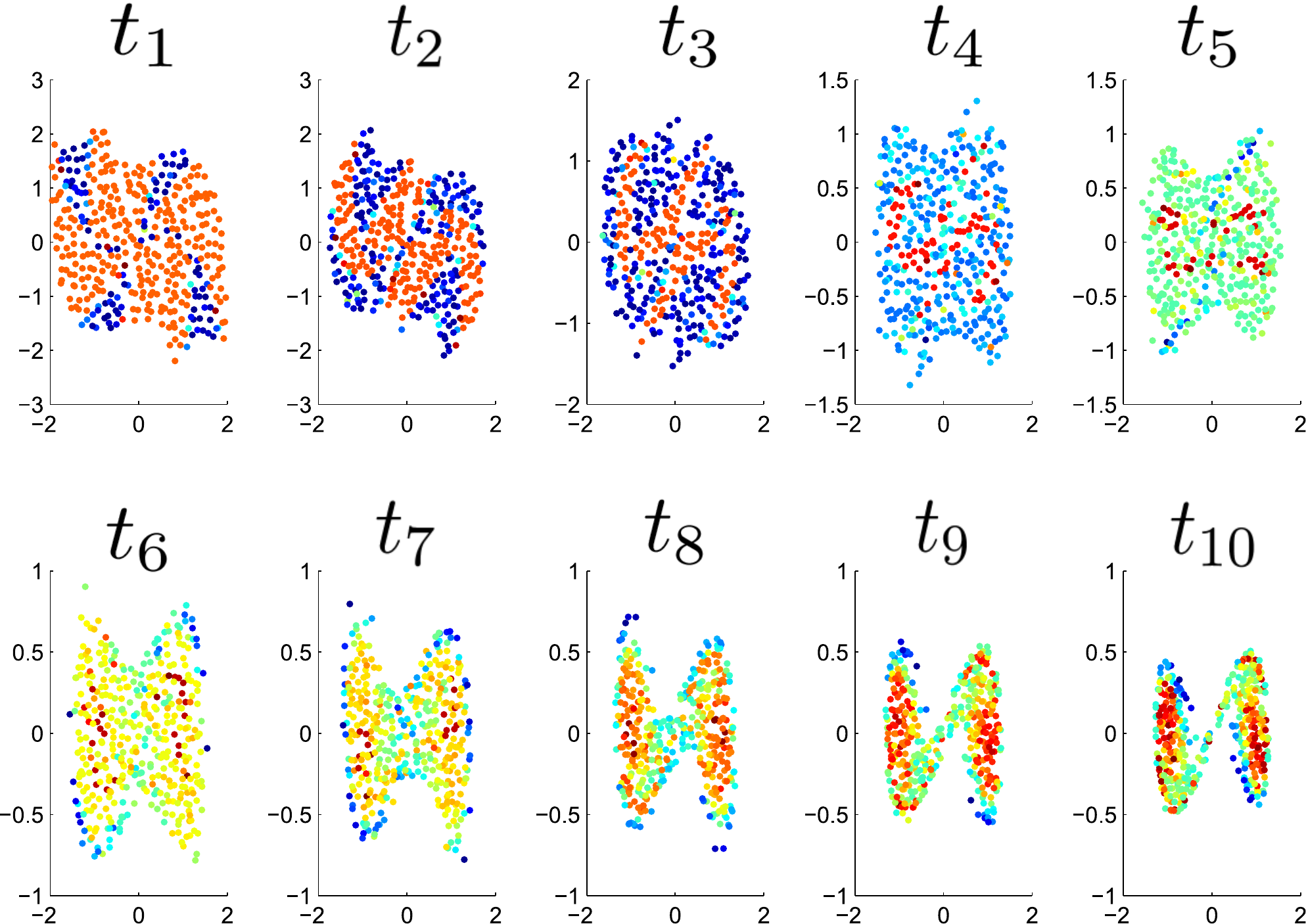}
 \end{center}
 \vspace*{-0.1in}
\caption{The distributional scattered data $\{t_{j}, \eta_{j}\}_{j=1}^{10}$ obtained by solving Liouville PDE for Duffing dynamics. The color value indicates the magnitude (\emph{red = high, blue = low}) of the joint PDF $\eta_{j}$. In our simulation, $t_{j} = \frac{j}{2}$, where $j=1,2,\hdots,10$.}
\label{DuffingScatteredData}
\vspace*{-0.1in}
\end{figure}

Let Fig. \ref{DuffingScatteredData} be the distributional data akin to Fig. \ref{OverallSchematic}, observed by the modeler. A continuous-time model is sought over each horizon: $t \in [t_{j}, t_{j+1})$. To solve this problem, we employ the Benamou-Brenier space-time optimization formulation described in Section II.D, resulting a vector field $v_{j}\left(x_{1}(t),x_{2}(t),t\right)$ per horizon, which solves the two point Liouville boundary value problem (guaranteeing end-point PDF matches) while incurring minimum amount of work over each $[t_{j}, t_{j+1})$. For this purpose, we take the two end point scattered data representation of $\eta_{j}$ and $\eta_{j+1}$, and interpolate the data over a regular grid, followed by Douglas-Rachford proximal operator splitting algorithm \cite{GabrielProxOpSplitArXiv2013} to solve the ensuing non-smooth convex optimization (\ref{BBcost})-(\ref{BenamouBrenierEulerian}), resulting the vector field $v_{j}\left(x_{1}(t),x_{2}(t),t\right)$. Fig. \ref{T1toT2} and \ref{T8toT9} show the gridded observed PDFs and the intermediate PDF reconstructions for $\left(t_{1},\eta_{1}\right) \rightarrow \left(t_{2},\eta_{2}\right)$, and $\left(t_{8},\eta_{8}\right) \rightarrow \left(t_{9},\eta_{9}\right)$, respectively, superimposed with their respective Benamou-Brenier vector fields (\emph{black arrows}). In Fig. \ref{TransportCompare}, we compare the PDF transportation paths for $t\in [t_{1},t_{2})$ in $W$, for the true Duffing dynamics (\ref{DuffingDynamics}) and the optimal transport dynamics. In view of Remark \ref{GradientFlow}, this plot shows that unlike the Brenier-Benamou \emph{gradient} vector field, (\ref{DuffingDynamics}) does not result into geodesic PDF transport. This is not surprising, since $\nabla \times f\left(x_{1},x_{2}\right) = \left(- 3\alpha x_{1}^{2} - \beta - 1\right) \widehat{\mathbf{k}} = - 3 x_{1}^{2}\:\widehat{\mathbf{k}}$, i.e. Duffing vector field has non-zero vorticity everywhere except $x_{1} = 0$, thus causing a clockwise rotational flow that requires more transportation effort than what could be achieved by a gradient flow.

\begin{figure}[t]
\begin{center}
 \includegraphics[width=0.49\textwidth]{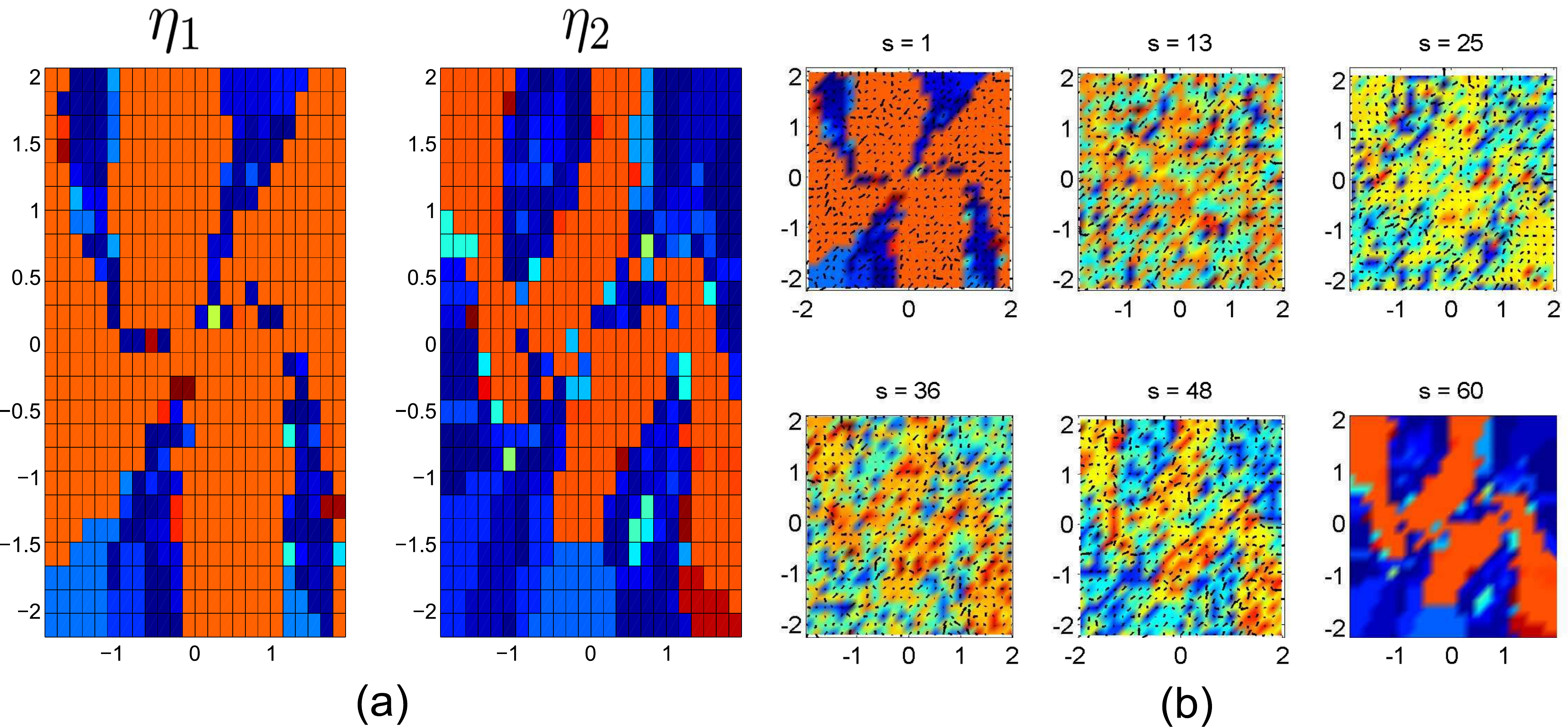}
 \end{center}
 \vspace*{-0.05in}
\caption{(a) The gridded PDFs $\eta_{1}$ and $\eta_{2}$; (b) The background color (\emph{red = high, blue = low}) shows optimal transport reconstructions for PDF $\eta\left(t\right)$, $t \in [t_{1}, t_{2})$, superimposed with Benamou-Brenier vector field $v_{1}^{\star}$ (\emph{black arrows}). The interval $[t_{1}, t_{2})$ was subdivided into 60 divisions, denoted by the index $s$ above, i.e. $s=0 \Leftrightarrow t_{1}$, $s=60 \Leftrightarrow t_{2}$. Notice that the vector field vanishes at $t_{2}$.}
\label{T1toT2}
\vspace*{-0.05in}
\end{figure}

\begin{figure}[t]
\begin{center}
 \includegraphics[width=0.49\textwidth]{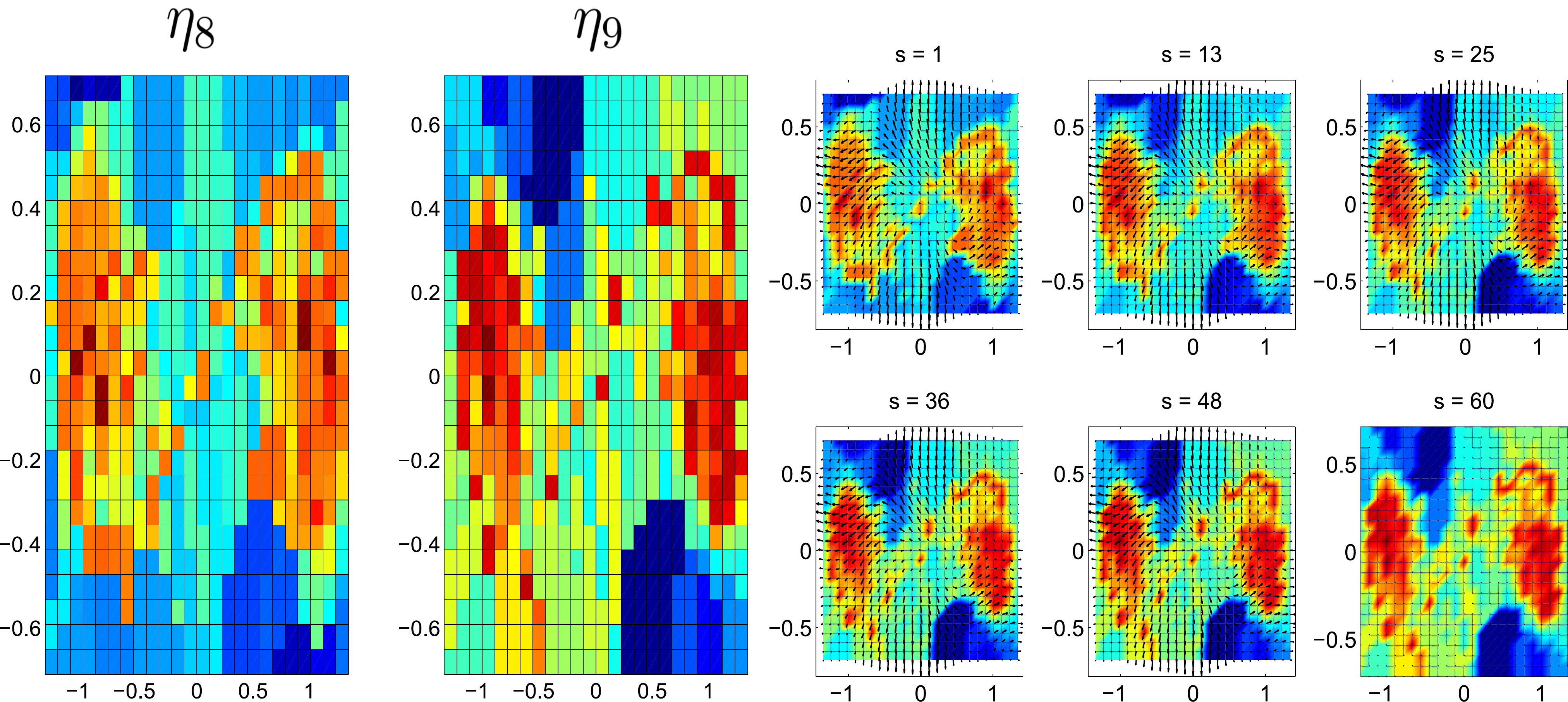}
 \end{center}
 \vspace*{-0.05in}
\caption{(a) The gridded PDFs $\eta_{8}$ and $\eta_{9}$; (b) The background color (\emph{red = high, blue = low}) shows optimal transport reconstructions for PDF $\eta\left(t\right)$, $t \in [t_{8}, t_{9})$, superimposed with Benamou-Brenier vector field $v_{8}^{\star}$ (\emph{black arrows}). Like Fig. \ref{T1toT2}, $s=0 \Leftrightarrow t_{8}$, $s=60 \Leftrightarrow t_{9}$. Again, the vector field vanishes at $t_{9}$.}
\label{T8toT9}
\vspace*{-0.05in}
\end{figure}

\begin{figure}[h]
\begin{center}
 \includegraphics[width=0.43\textwidth]{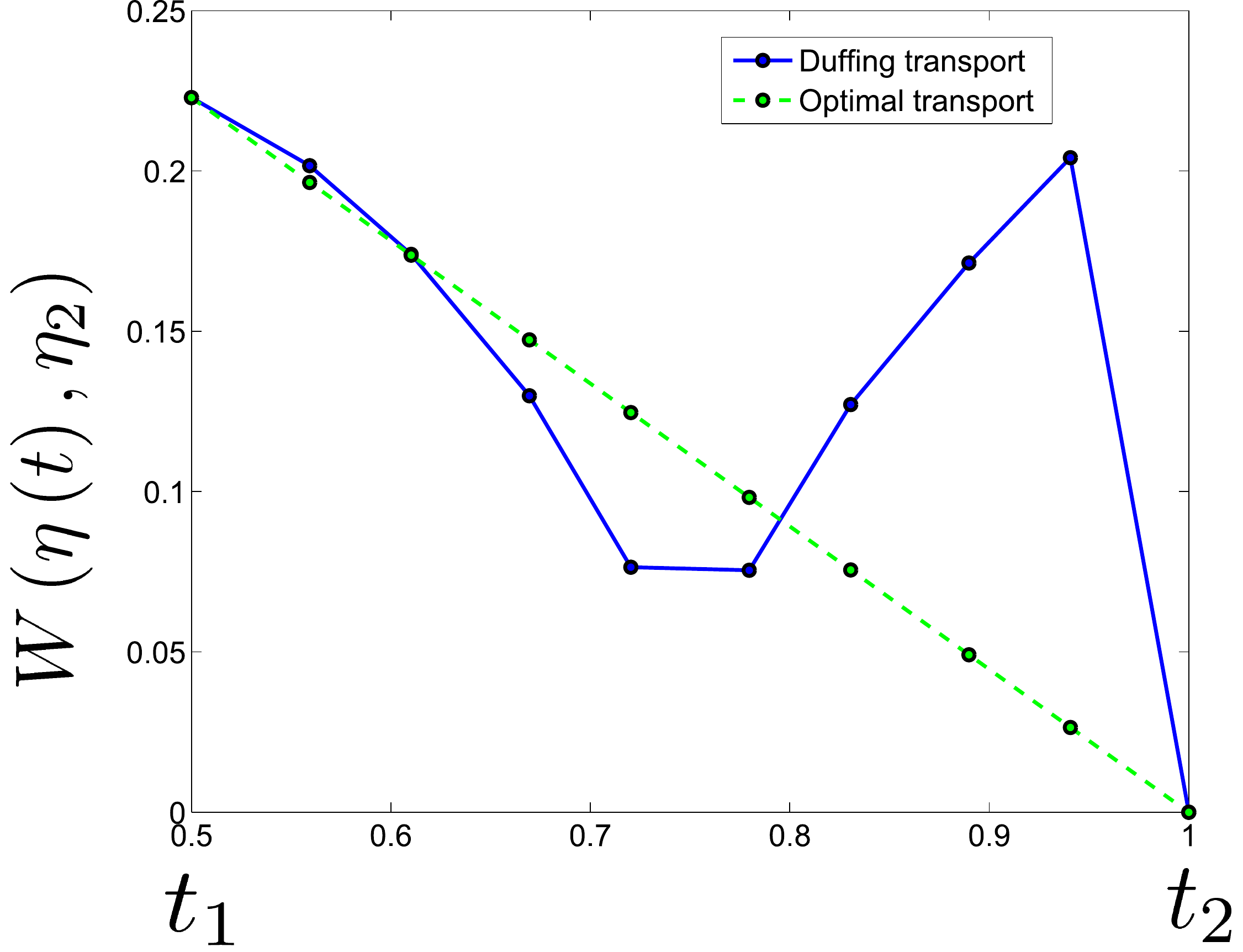}
 \end{center}
 \vspace*{-0.1in}
\caption{Comparison of optimal transport (Benamou-Brenier dynamics) with Duffing transport (true dynamics) for $t \in [t_{1}, t_{2})$.}
\label{TransportCompare}
\vspace*{-0.1in}
\end{figure}

\end{example}


\section{Model Refinement}
In this Section, we consider \textbf{Problem 1.3}, namely refining a baseline model against experimental data. We first formulate the model refinement problem as that of finding the optimal transport map introduced in Section II.

\subsection{General formulation}
We formulate the model refinement problem (Fig. \ref{ModelRefinementBlockDiagm}) as the natural successor of the distributional model validation formulation proposed in \cite{HalderBhattacharya2011, HalderBhattacharya2012}. In the validation problem, the model predicted output PDF $\widehat{\eta}$ is compared with the experimentally observed output PDF $\eta$, at each instance of measurement availability $t_{j}$, $j=0,1,\hdots,M$, and an inference is made by looking at the prediction-observation gap quantified via $W\left(t_{j}\right)$. The key insight behind our refinement formulation is that usually there is no specific requirement on the structure of the refined model, as long as we can make the refined dynamics track the observed output PDFs. This provides us the freedom to formulate the model refinement problem over the model's output map while keeping the model's state equation intact. This has two implications: (i) the refinement algorithm will involve the output dimension $d$, typically less than the state dimension, and (ii) both state and output modeling errors would be accounted by updating the model's output map. To make the ideas precise, we give the model refinement problem statement for a model whose output map is given by $\widehat{y} = \widehat{h}\left(\widehat{x}\right)$, where $\widehat{x}$ and $\widehat{y}$ are model-predicted state and output vectors.

\begin{figure}[t]
\begin{center}
 \includegraphics[scale=0.67]{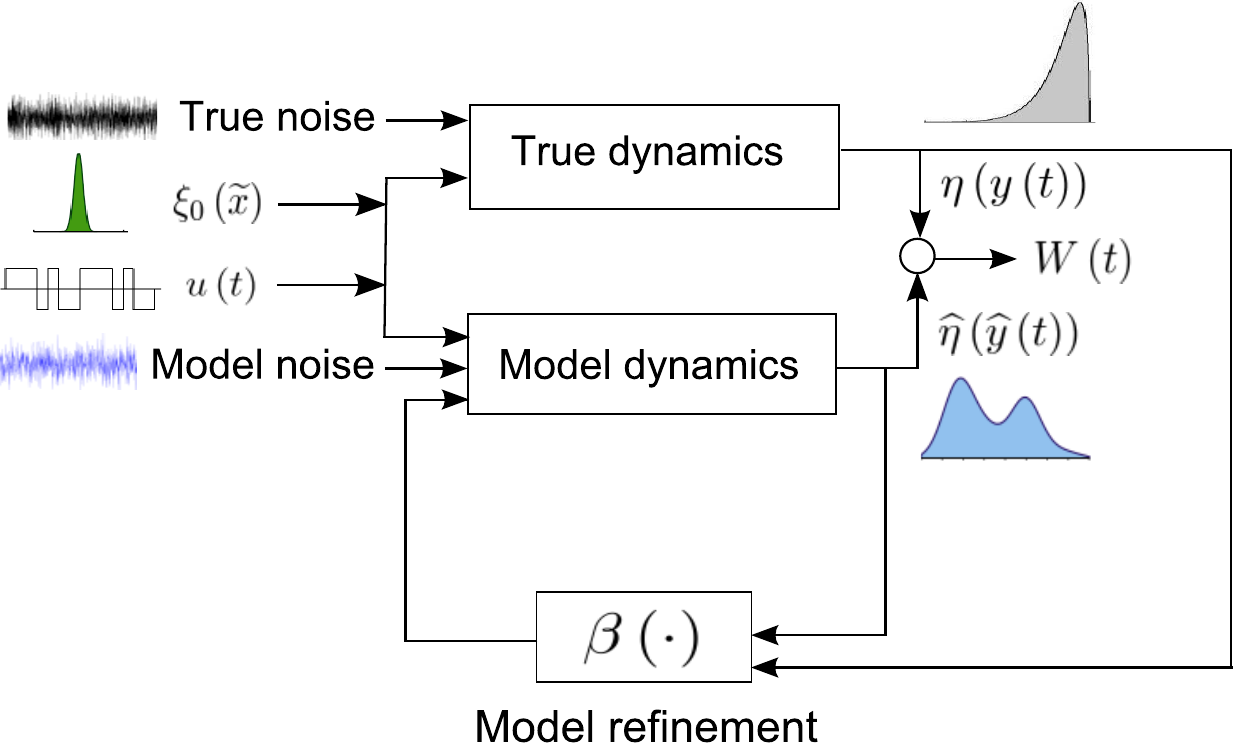}
 \end{center}
 \vspace*{-0.1in}
\caption{The block diagram for proposed model refinement formulation. Here $\xi_{0}\left(\widetilde{x}\right)$ refers to the joint PDF supported on the space of initial conditions and parameters, $u(t)$ is an open-loop control command, and the true and model dynamics can be affected by different noises.}
\label{ModelRefinementBlockDiagm}
\vspace*{-0.1in}
\end{figure}
\vspace*{-0.05in}

\subsection{Problem Statement: Transcribing Model Refinement Problem as Finding Optimal Transport Map}
At $t = t_{j}$, let us introduce $\widehat{y}_{j}^{-} \triangleq \widehat{y}_{j}$, and denote $\widehat{h}^{-}\left(\cdot\right) \triangleq \widehat{h}\left(\cdot\right)$. We want to find the \emph{Brenier map} $\beta_{j}\left(\cdot\right)$ for updating the predicted output, i.e. $\widehat{y}_{j}^{+} = \beta_{j}\left(\widehat{y}_{j}^{-}\right)$, where $\widehat{y}_{j}^{+} \sim \eta_{j}$ and $\widehat{y}_{j}^{-} \sim \widehat{\eta}_{j}$. In other words, find $\beta_{j}\left(\cdot\right)$ such that $\eta_{j} = \beta_{j} \:\sharp\: \widehat{\eta}_{j}$. Clearly, this problem is underdetermined since there are many ways to morph $\widehat{\eta}_{j}$ to $\eta_{j}$. Then we must look for an \emph{optimal push-forward map} $\beta_{j}^{\star}\left(\cdot\right)$ that would require minimum amount of transport effort among all possible push-forward maps $\beta_{j}\left(\cdot\right)$, i.e. we solve (\ref{BetaVariational}). Once $\beta_{j}^{\star}\left(\cdot\right)$ has been found, the refined model is given by augmenting the model's state equation with the new output map:
\begin{eqnarray}
\widehat{y}_{j} = \beta_{j}^{\star} \circ \widehat{h}\left(\widehat{x}\right).
\label{RefinedModelBrenierMap}
\end{eqnarray}

\begin{example}
\textbf{(Refining Linear Model against Gaussian Measurements)} Let the true data being generated by the discrete-time LTI system $x_{j+1} = A x_{j}$, $y_{j} = C x_{j}$, that is unknown to the modeler. The proposed model is $\widehat{x}_{j+1} = \widehat{A} \widehat{x}_{j}$, $\widehat{y}_{j} = \widehat{C}\widehat{x}_{j}$, where the Schur-Cohn stable matrices $A$ and $\widehat{A}$ are given by
\begin{eqnarray}
A = \begin{bmatrix} 0.4 & & -0.1\\ 2 & & 0.6\end{bmatrix}, \quad \widehat{A} = \begin{bmatrix} 0.2 & & -0.7\\ -0.7 & & 0.1\end{bmatrix},
\end{eqnarray}
and the output matrices are
\begin{eqnarray}
 C= \begin{bmatrix} -1 & & 0.03\\ -0.2 & & 0.8\end{bmatrix}, \quad \widehat{C} = \begin{bmatrix} 1 & & 0\\ 0 & & 1\end{bmatrix}.
\end{eqnarray}
Starting from the initial Gaussian state PDF $\xi_{0} = \mathcal{N}\left(\mu_{0},P_{0}\right)$ with $\mu_{0} = \{1, \, 3\}^{\top}$, $P_{0} = \begin{bmatrix} 10 & & 6\\ 6 & & 7\end{bmatrix}$, we refine the model at three instances of measurement availability: $j = 1, 2$, and 3. The results for the model refinement algorithm are shown in Fig. \ref{Gauss2Gauss}. To illustrate how the results of Section II.E are applied in this particular refinement problem, we provide the following Theorem.
\end{example}

\begin{figure}[tb]
\begin{center}
 \includegraphics[scale=0.41]{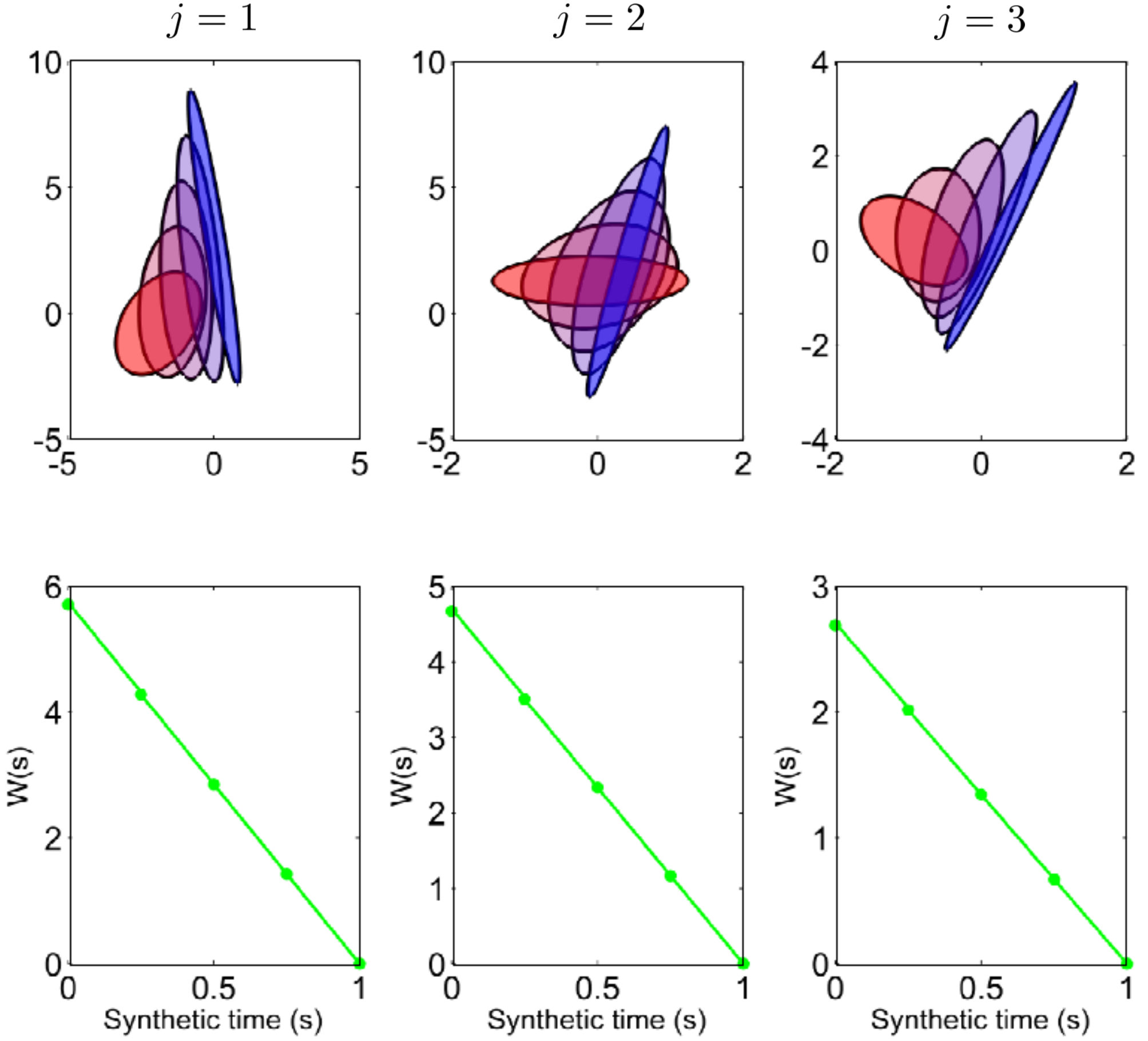}
 \end{center}
 \vspace*{-0.1in}
\caption{Shown here is the refinement process of the linear model $\left(\widehat{A},\widehat{C}\right)$ at times $j = 1, 2$ and 3, so that the model predicted output PDFs match with the true PDF, generated by $\left(A,C\right)$ at each $j$. Since both processes are Gaussian, the \emph{top row} shows 1-$\sigma$ ellipses of the respective normal PDFs (\emph{red = model predicted}, \emph{blue = true}). At every fixed $j$, we also plot intermediate Gaussians generated during the refinement process. The color of these intermediate 1-$\sigma$ ellipses are interpolated \emph{from red to blue}, to show the notion of synthetic time $s \in \left[0,1\right]$, as the physical time index $j$ remains zero-order hold. This also shows that the Gaussian-to-Gaussian refinement happens via Gaussians, i.e. the set of Gaussian PDFs is geodesically convex. The \emph{bottom row} shows that although the PDFs over synthetic time gets nonlinearly interpolated (McCann's \emph{displacement interpolation} \cite{McCann1997}), the Wasserstein distance $W(s)$ gets linearly interpolated, as predicted by (\ref{WIsLinearInterpolation1}).}
\label{Gauss2Gauss}
\vspace*{-0.1in}
\end{figure}
\vspace*{-0.05in}

\begin{theorem}\textbf{($s$\textsuperscript{th} synthetic time PDF at $j$\textsuperscript{th} physical time)}
Let $s \in \left[0,1\right]$ and consider the above linear Gaussian refinement problem with initial PDF $\mathcal{N}\left(\mu_{0},P_{0}\right)$. At the $j$\textsuperscript{th} instance of measurement availability, the intermediate PDF during refinement is a Gaussian PDF $\mathcal{N}\left(\mu_{\widehat{y}\rightarrow y}\left(s\right), \Sigma_{\widehat{y}\rightarrow y}\left(s\right)\right)$ where
\begin{eqnarray}
\mu_{\widehat{y}\rightarrow y}\left(s\right) = \left[\left(1 - s\right) \:\widehat{C}\widehat{A}^{j} \: + \: s \:CA^{j}\right] \mu_{0},
\label{sthIntermediateAtjthTimeMean}
\end{eqnarray}
{\small{\begin{eqnarray}
\Sigma_{\widehat{y}\rightarrow y}\left(s\right) =\left[\left(1 - s\right) \:I \: + \: s \:\Gamma\left(j\right)\right] \left(\left(\widehat{C}\widehat{A}^{j}\right) P_{0} \left(\widehat{C}\widehat{A}^{j}\right)^{\top}\right) \nonumber\\
\left[\left(1 - s\right) \:I \: + \: s \:\Gamma\left(j\right)\right],
\label{sthIntermediateAtjthTimeCov}
\end{eqnarray}}}
where
{\small{\begin{eqnarray}
\Gamma\left(j\right) \triangleq \sqrt{\left(CA^{j}\right) P_{0} \left(CA^{j}\right)^{\top}} \: \left(\sqrt{\left(CA^{j}\right) P_{0} \left(CA^{j}\right)^{\top}} \right. \nonumber\\
 \left. \left(\widehat{C}\widehat{A}^{j}\right) P_{0} \left(\widehat{C}\widehat{A}^{j}\right)^{\top} \: \sqrt{\left(CA^{j}\right) P_{0} \left(CA^{j}\right)^{\top}}\right)^{-1/2} \: \nonumber\\
 \sqrt{\left(CA^{j}\right) P_{0} \left(CA^{j}\right)^{\top}}.
\end{eqnarray}}}
\end{theorem}

\begin{proof}
We know that $\mu_{y}\left(j\right) = C \mu_{x}\left(j\right) = C A^{j} \mu_{0}$, and similarly, $\widehat{\mu}_{\widehat{y}}\left(j\right) = \widehat{C} \widehat{A}^{j} \mu_{0}$. On the other hand, we have ${\small{\Sigma_{y}\left(j\right) = C \Sigma_{x}\left(j\right) C^{\top} = C A^{j}P_{0}\left(A^{j}\right)^{\top} C^{\top}}}$, and similarly, ${\small{\widehat{\Sigma}_{\widehat{y}}\left(j\right) = \left(\widehat{C} \widehat{A}^{j}\right)P_{0}\left(\widehat{C} \widehat{A}^{j}\right)^{\top}}}$.

From (\ref{PDFIsNonlinearInterpolation1}), we get $\mu_{\widehat{y}\rightarrow y}\left(s\right) = \left(1-s\right)\widehat{\mu}_{\widehat{y}}\left(j\right) + s \mu_{y}\left(j\right) = \left[\left(1 - s\right) \:\widehat{C}\widehat{A}^{j} \: + \: s \:CA^{j}\right] \mu_{0}$. Similarly, $\Sigma_{\widehat{y}\rightarrow y}\left(s\right) =\left[\left(1 - s\right) \:I \: + \: s \:\Gamma\left(j\right)\right] \widehat{\Sigma}_{\widehat{y}}(j) \left[\left(1 - s\right) \:I \: + \: s \:\Gamma\left(j\right)\right]$, where $\Gamma(j) = \sqrt{\Sigma_{y}(j)} \left(\sqrt{\Sigma_{y}(j)}\:\widehat{\Sigma}_{\widehat{y}}(j)\:\sqrt{\Sigma_{y}(j)}\right)^{-\frac{1}{2}} \sqrt{\Sigma_{y}(j)}$ (from (\ref{BrenierMapGauss2GaussMatrix})). Substituting the covariance matrix formulae in terms of the respective system matrices, the result follows.
\end{proof}


\section{Conclusions}
In this paper, we argued that many problems in dynamic data driven modeling lead to distributional observation, \emph{either} in natural stochastic sense, \emph{or} in the sense of concentration. We showed how the optimal transport theory can offer a disciplined approach to solve such problems like finite horizon feedback control of PDFs, data-driven reduced-order modeling, and refining a baseline model.


\end{document}